\documentclass[aps,twocolumn,amssymb]{revtex4}

\newcommand{\iec}{\mbox{i.\,e.\,}}

\newcommand{\egc}{\mbox{e.\,g.\,}}

\newcommand{\pb}[2]{\ensuremath{\left\{ #1 , #2 \right\} }}
\newcommand{\dr}[1]{\ensuremath{\mathrm{d} #1\,}}
\newcommand{\mc}[1]{\ensuremath{\mathcal{#1}}}
\newcommand{\edf}{\ensuremath{=_{_{df}}}}

\newcommand{\ket}[1]{\ensuremath{\left|  #1 \right\rangle}}
\newcommand{\bra}[1]{\ensuremath{\left\langle #1 \right|}}
\newcommand{\bk}[2]{\ensuremath{\left\langle #1 | #2 \right\rangle}}

\newcommand{\matel}[3]{\ensuremath{\bra{#1} #2 \ket{#3}}}
\newcommand{\op}[1]{\ensuremath{\widehat{\textsf{\ensuremath{#1}}}}}

\newcommand{\tr}{\textsf{Tr}}

\newcommand{\be}{\begin{equation}}
\newcommand{\ee}{\end{equation}}
\newcommand{\e}[1]{\mathrm{e}^{#1}}

\newtheorem{theorem}{Theorem}
\newenvironment{proof}[1][Proof]{\begin{trivlist}
\item[\hskip \labelsep {\bfseries #1}]}{\end{trivlist}}

\begin{document}

\title{Recurrence Theorems: a unified account}
\author{David Wallace}
\email{david.wallace@balliol.ox.ac.uk}
\affiliation{Balliol College, University of Oxford}
\date{\today}
\begin{abstract}
I discuss classical and quantum recurrence theorems in a unified manner, treating both as generalisations of the fact that a system with a finite state space only has so many places to go. Along the way I prove versions of the recurrence theorem applicable to dynamics on linear and metric spaces, and make some comments about applications of the classical recurrence theorem in the foundations of statistical mechanics.
\end{abstract}

\maketitle

\section{Introduction}

The Poincar\'{e} recurrence theorem plays an important role in the foundations of statistical mechanics, dating back to Zermelo's original objection to Boltzmann's H theorem.\footnote{See \cite{brownboltzmann}, and references therein} The theorem exists in both classical and 
quantum  forms, but only the classical version is widely discussed. This is unfortunate, as there are foundationally important differences between the two.
Furthermore, in proofs of the classical theorem measure-theoretic technicalities can obscure the very elementary idea underpinning the theorem, and standard proofs of the quantum version obscure the links between classical and quantum versions.

In this paper I provide a unified treatment of recurrence, beginning with the toy example of a system with discretely many states and based around the intuition that if a system only has finitely many places to go it will eventually enter one such place twice. In the discrete case this is literally how the theorem works; I generalise it to dynamics on spaces equipped with notions of conserved volume, of conserved inner product, and (in the Appendix) of conserved length, showing that in each case the conserved quantity extends the basic idea from discrete to continuous state spaces. I illustrate how recurrence theorems for classical and quantum mechanics, respectively, are instances of the first two extensions. I stress the distinction between recurrence (every point recurs after some time interval) and \emph{uniform} recurrence (there is some time interval after which every point recurs). On the basis of these results, I observe that two properties of recurrence that have been frequently discussed in the literature are just artifacts of classical mechanics.

\section{Warm-up: recurrence in finite systems}

At a very abstract level, we can define a \emph{dynamical system} as specified by:
\begin{itemize}
\item A set \mc{S} of \emph{states}
\item An evolution rule $U$, which maps $\mc{S}$ to itself.
\end{itemize}
The idea is that $U$ maps each state to the state into which it will evolve after 1 time-step. We can then define a time evolution operator $U(n)$ as follows: $U(1)=U$, $U(2)=U\cdot U$, $U(3)=U\cdot U \cdot U$, etc. (or, more formally: $U(1)=U$, $U(n+1)=U\cdot U(n)$), satisfying $U(n+m)=U(n)U(m)$). Although in this model of dynamics time is discrete, this is no real limitation: given a continuous-time evolution operator $V(t)$, we can pick some arbitrarily short time $\tau$ and define $U=V(\tau)$; then $U(n)\edf V(n\tau)$.

The idea of recurrence is that every state, as it evolves forward in time, in some sense eventually returns to its original state. If the state space $\mc{S}$ has finitely many points, we can interpret this entirely literally. Suppose we say that:
\begin{itemize}
\item A dynamical system is exactly recurrent if for every $s\in \mc{S}$, there is some $n_s>0$ such that $U(n_s)s=s$.
\item A dynamical system is invertible if $U$ is a one-to-one map: that is, if for every $s\in \mc{S}$ there is a unique $t\in \mc{S}$ such that $U(t)=s$. (From this it follows that $U(n)$ can be extended uniquely to negative $n$ while preserving $U(n+m)=U(n)U(m)$.
\item A dynamical system is finite if its state space $\mc{S}$ contains  finitely many points.
\end{itemize}
Then we can easily prove the
\begin{theorem}
\textbf{(Finite Recurrence Theorem)}
Any finite invertible system is exactly recurrent.
\end{theorem}
\begin{proof}
Let $\mc{S}$ contain $N$ members. Then for any given $s$, not all of 
\[s, U(1)s, U(2)s, \ldots U(N)s\] can be distinct, so there must be distinct $n,m$ with $n<m$ such that $U(n)s=U(m)s$. Then $s=U(-n)U(n)s=U(-n)U(m)s=U(m-n)s$. $\Box$
\end{proof}
Informally: since there are only finitely many states, the evolution of $s$ must eventually pass through the same state twice; because the dynamics is given by a one-to-one map, the \emph{first} state to be entered twice must be the original state, so that no two states are mapped to the same state.

Note that both requirements are needed. If the system is not finite, it can continue exploring new states forever; if it is not invertible, it can get stuck at some state or set of states and never get out again.

Now consider the space $\mc{U}(\mc{S})$ of \emph{all} maps from \mc{S} to itself. Given some evolution operator $U$, we can define a function $L_U$ from $\mc{U}(\mc{S})$ to itself by
\be 
L_U(V)=U\cdot V.
\ee
Formally, we can regard $\mc{U}(\mc{S})$ as a state space, and $L_U$ as an evolution operator on that state space, so as to create a new dynamical system. If $\mc{S}$ is finite, so is $\mc{U}(\mc{S})$ (there are $N^N$ maps from an $N$-member set to itself); if $U$ is invertible, so is $L_U$ $(L_{U^{-1}}\cdot L_U (V)=U^{-1}\cdot U \cdot V=V$). So in fact, this dynamical system is also recurrent. In particular, if $1$ is the identity map, then there exists some $n$ such that
\be 
(L_U)^n(1)\equiv U(n) \cdot 1 =U(n)=1.
\ee
But if $U(n)=1$, then $U(n)s=s$ for all $s\in \mc{S}$: that is, there is some \emph{fixed} time $n$ such that after $n$ time-steps, every state has been time-evolved back to itself Let us call this stronger property \emph{uniform recurrence}: formally, a dynamical system is  \emph{uniformly exactly recurrent} if for some fixed $n>0$, \emph{all} states $s\in\mc{S}$ satisfy $U(n)s=s$. 

We have now proved the
\begin{theorem}
\textbf{(Finite Uniform Recurrence Theorem)}
Any finite invertible system is uniformly exactly recurrent.
\end{theorem}

In essence, all recurrence results we will consider have this form: recurrence occurs just because the system only has so many places to go, and runs out of them, and because its one-to-one dynamics means that this can only happen if it returns to its original state. The problem is that in most cases of physical interest, the state space has infinitely many points in it, so simple arguments based on state-counting will not work. We need instead to find some precise sense in which the system has only finitely many \emph{interestingly different} places to go, so that it returns to some state \emph{virtually the same as} its initial state. This will be our task for the rest of the paper.

\section{Recurrence in volume}

Suppose that the state space $\mc{S}$ of a given dynamical system is equipped with a \emph{measure}, a rule that defines a notion of volume on the subsets of \mc{S}. (Informally, a measure $\mu$ is a rule that associates to a subset $S\subset \mc{S}$ a value $\mu(S)$ which is a nonnegative real number or $\infty$, such that the volume of a union of disjoint sets is the sum of the volumes of the individual sets and the volume of the empty set is zero; for a more formal definition, see \cite{rudinrealcomplex}. Not all subsets can be assigned a measure (there are pathological `non-measureable' sets) but any remotely simple set is measureable.) We define a \emph{measureable dynamical system} as a dynamical system equipped with a measure in this way. We can now use this notion to define an appropriate generalisation of recurrence.
\begin{itemize}
\item A measureable dynamical system is \emph{recurrent in volume} iff for any measureable $S\subset \mc{S}$, the set $T\subset S$ of points in $S$ that do not eventually re-enter $S$ under dynamical evolution has measure zero.
\end{itemize}
To put this more informally: pick any state, and choose any non-zero-volume region containing that state. If the system is recurrent in volume, virtually all points in that region will eventually return to it --- ``virtually'' in the sense that the set of those which do not has measure zero. By taking the region arbitrarily small, we can ensure that (in some sense) almost all states arbitrarily close to our original state will return arbitrarily close to it.

To prove recurrence in volume, we need assumptions analogous to finiteness and invertibility. They are:
\begin{itemize}
\item A measureable dynamical system is \emph{finite in measure} iff $\mu(\mc{S})<\infty$.
\item A measureable dynamical system is \emph{volume-preserving} iff $\mu(U(S))=\mu(S)$ for all measureable $S$.
\end{itemize}
We can then prove the
\begin{theorem}
\textbf{(Volume Recurrence Theorem)}
Any measureable dynamical system that is finite and volume-preserving is recurrent in volume.
\end{theorem}
\begin{proof}
Pick a measureable subset $S\subset\mc{S}$, and let $T\subset S$ be the set of points in $S$ that never return to $S$, \iec the points $s\in \mc{S}$ which satisfy $U(n)s\notin S$ for all $n>0$. Now define $T_n=U(n)T$. Since the dynamics is measure-preserving, for any $n,m$ (with $n<m$) we have $\mu(T_n \cap T_m)=\mu(T \cap T_{m-n})=0$, \iec the overlap of any two $T_n$ has measure zero. So if we define $T'_m$ as the set of points in $T_m$ but not in any of the $T_n$ with $n<m$, then $\mu(T'_n)=\mu(T_n)$. 

The $T'_n$ are mutually disjoint, so if $\mc{T}=\cup_n T'_n$ (and so also $\mc{T}=\cup_n T_n$), $\mu(\mc{T})=\mu(T')+\mu(T'_1)+\mu(T'_2)+\cdots$.

But since the dynamics is volume-preserving, $\mu(T)=\mu(T_i)=\mu(T'_i)$ for all $i$. There are two ways of satisfying this requirement: 
\begin{itemize}
\item[(i)]$\mu(\mc{T})=\infty.$
\item[(ii)]$\mu(T)=0$.
\end{itemize}
But (i) is ruled out by the assumption that the system is finite in measure, so (ii) is the only possibility. $\Box$
\end{proof}

How can we define a volume-based version of uniform recurrence? It's not entirely clear, but it's moot in any case, because no such result is true. Consider, for instance, the following dynamical system: the state space is a cylinder of radius 1 and height 1, so that states are labelled by pairs $(z,\theta)$ with $z\in [0,1]$ and $\theta\in[0,2\pi)$. If the dynamical rule is $(z,\theta)\rightarrow (z,[\theta+z]\,\mathrm{ modulo }\, 2\pi)$ (i.e., if it rotates each point an amount proportional to $z$) then the ordinary volume of this cylinder provides a preserved measure, so that the system is a finite, invertible, measure-preserving dynamical system. The volume recurrence theorem holds, and indeed this is easy to see informally: for any $z$, and any $\epsilon$, we can find $n$ so that $n z$ is within $\epsilon$ of being an integer multiple of $2\pi$.

If the system were \emph{uniformly} recurrent, though, we would expect (e.g.) the line of points $(z,0)$ for arbitrary $z$ to return arbitrarily close to itself. But it's clear that this won't happen: the line will wind more and more tightly around the cylinder, and never unwind. 

To get further mathematical insight into why this happens, consider that the space of measure-preserving functions on a (non-finite) space is typically very large --- so large, in fact, that there is no natural way to define a volume measure on it. So the analogue of our finite-case strategy, where we considered the space of maps itself as a dynamical system, will not work in this case.

We will shortly see that there are alternative (and physically more relevant) generalisations of recurrence in which uniform recurrence \emph{does} hold. Firstly, though, we should stop and consider the most important application of volume recurrence.

\section{Application: classical mechanics}

Classical mechanics in the Hamiltonian formalism (for given time-step $\tau$, and assuming finitely many degrees of freedom) is a dynamical system: the state space is phase space, the evolution rule is evolution for time $\tau$ under Hamilton's equations, and the volume is given by\footnote{More formally, this is the volume defined by the symplectic structure $\omega$ on phase space: $\mu(V)=\int_V \omega\wedge \omega \cdots \wedge \omega$, where $\omega\wedge \omega \cdots \wedge \omega$ is the $n$-fold exterior product of $\omega$ with itself.} 
\be 
\mu(V)=\int_V\dr{q^1}\dr{p_1} \cdots \dr{q^n} \dr{p_n}
\ee
By Liouville's theorem, this is conserved under Hamiltonian evolution; as such, classical mechanics is an invertible, measure-preserving, measureable dynamical system. 

It is not a finite system (phase space, in general, has infinite volume, since momentum can increase without limit). However, Hamiltonian dynamics conserves energy, and it is entirely possible for the subset of phase space corresponding to energies between $E$ and $\delta E$ (for arbitrarily small $\delta E$) to have finite volume.\footnote{With only slightly more care, we can prove a version of the recurrence theorem that applies if \emph{any given energy hypersurface} is finite in volume. However, as there can be no physically interesting distinction between the two cases, even this slight increase in care seems unnecessary.} This gives us a form of the 
\begin{theorem}
\textbf{(Poincar\'{e} recurrence theorem)} For any classical-mechanical system with finite-dimensional phase space, if the subset of phase space  with energies lying in some band $(E,E+\delta E)$ has finite Liouville volume then the restriction of the dynamics to that subset is recurrent in volume.
\end{theorem}
As an important example of this, consider a system whose Hamiltonian has the form
\be 
H(q^1,p_1,\ldots q^n,p_n)=\sum_{i,j}C^{ij}p_i p_j + V(q^1,\ldots q^n),
\ee
where $V$ is bounded below, $C^{ij}$ is positive definite, and the positions are confined to a finite region. Then the momenta will also be confined to a finite region for given energy, and the Poincar\'{e} recurrence theorem will hold. (For a case where the positions are not confined, consider a particle moving freely in empty space (clearly this will not recur); for a case where $V$ is not bounded below, consider Newtonian gravity.)

\section{Recurrence in linear dynamical theories}

We say that a dynamical theory is \emph{linear} if its state space is some subspace of an inner product space (whose inner product we write $\langle \cdot , \cdot \rangle$). Provided the space is of some finite dimension $N$, we can use this inner product to define a volume on the space, informally\footnote{Formally (at least if the state space is an open subset of the inner product space) pick an orthonormal basis; use this basis to identify the state space with a subset of $R^N$; use this identification to carry the Lebesgue measure over to the state space.} just by defining the volume of a cube of side $\lambda$ as $\lambda^N$. Also given the inner product, we can define the distance $d(x,y)$ between two states $x$, $y$ as the length of the vector between them: that is, $d(x,y)=\sqrt{\langle x-y,x-y\rangle}$. This function is sometimes called the \emph{metric}.

We now define a linear dynamical theory as 
\begin{itemize}
\item \emph{bounded} if for some fixed $D$, $d(s,t)<D$ for all states $s,t$. 
\item \emph{inner-product-preserving} if for all states $s,t$, $\langle U(s),U(t)\rangle=\langle s,t\rangle$.
\end{itemize}
We then have 
\begin{theorem}
\textbf{(Linear recurrence theorem, preliminary form)} Any linear dynamical theory which is finite-dimensional, bounded, invertible and inner-product-preserving is recurrent in volume with respect to the volume function defined by the inner product.
\end{theorem}
\begin{proof}
If the system is inner-product-preserving, its dynamics preserves the volume measure. If the system is bounded, (say, with $d(s,t)<D$ for some $D$ and all states $s,t$), then the state space is contained within a cube of side $D$ and hence volume $D^N$, and so has finite volume. The result is now a corollary of the volume recurrence theorem. $\Box$
\end{proof}

Since quantum-mechanical systems are linear systems (taking the state space to be Hilbert-space states of norm 1), this result applies in particular to finite-dimensional quantum systems, and indeed can be extended to infinite-dimensional ones under certain conditions (as we will see). But this route to quantum recurrence, although it has the virtue of hewing closely to the classical model, understates the strength of recurrence that actually occurs in quantum theory --- essentially because preservation of an inner product is a much stronger condition than preservation of a volume.

Indeed, the mere presence of an inner product allows us to define an alternative concept of recurrence: we can consider a system to be recurrent if any state returns arbitrarily close to its starting point with respect to this metric. More precisely,
\begin{itemize}
\item A linear dynamical system is \emph{recurrent in metric} iff for any state $s$, and any $\epsilon>0$, there is some $n_s$ such that $d(U(n_s)s,s)<\epsilon$.
\item A linear dynamical system is \emph{uniformly recurrent in metric} iff for any $\epsilon>0$, there is some $n$ such that for all states $s$, $d(U(n)s,s)<\epsilon$.
\end{itemize}

We can then prove a stronger result:
\begin{theorem}
\textbf{(Linear recurrence theorem)} Any linear dynamical theory which is finite-dimensional, bounded, and inner-product-preserving is uniformly recurrent in metric.
\end{theorem}
\begin{proof}
We begin by proving ordinary recurrence. Suppose that the dynamical theory has a state space $\mc{S}$ that is a subset of an $N$-dimensional vector space \mc{V}, and pick some orthonormal set of vectors $v_1, \ldots v_n$. For any $N$-tuple of real\footnote{We're not really assuming that the space is a \emph{real} rather than complex space here; an $N$-dimensional complex vector space is also a $2N$-dimensional real vector space.} numbers $x^1,\ldots x^n$, and any $\epsilon>0$, we can construct the ball 
\be B_\epsilon(x^1,\ldots x^n)=\{x\in \mc{V}:d(x,\sum_i x^iv_i)<\epsilon/2\}.\ee
If $\mc{B}$ is the set of such balls defined by $N$-tuples $n_1 \epsilon/2\sqrt{N}, \ldots n_N\epsilon/2\sqrt{N}$ for arbitrary (positive or negative) integers $n_1, \ldots n_N$, it is easy to verify that 
\begin{itemize}
\item[(i)] every point in \mc{V} is in some element of \mc{B}; 
\item[(ii)] if \mc{S} is bounded, there is some \emph{finite} subset $B_1,\ldots B_M$ such that every point in $\mc{V}$ is in one of the $B_i$. 
\end{itemize}
Now pick an arbitrary state $s$, and consider its dynamical evolutes \[U(1)s, \,U(2)s,\, \ldots.\] Each one must lie in (at least) one of the $B_i$ (we can pick one arbitrarily if it lies in more than one) and so, since there are only finitely many of the $B_i$, we must be able to find $n,m$ ($n<m$) such that $U(n)s$ and $U(m)s$ both lie within the same $B_i$.

But if so, $d(U(n)s,U(m)s)<\epsilon$ ($B_i$ has radius $\epsilon/2$, so diameter $\epsilon$) and so, since the dynamics preserves $d$, $d(s,U(m-n)s)<\epsilon$: that is, the system is recurrent in metric.

We prove \emph{uniform} recurrence by the same trick used in the finite case: we construe the space of evolution operators as a dynamical system in its own right. Specifically: consider the space $\mc{U}(\mc{V})$ of linear transformations of \mc{V}. This too is a vector space, and it can be equipped with an inner product as follows:
\be 
\langle U,V\rangle=\tr (U^\dagger V).
\ee
Any norm-preserving $V$ satisfies $\langle V,V\rangle=\tr(V^\dagger V)=\tr(1)=N$, so that the set of all such $U$ is a bounded subset of $\mc{U}(\mc{V})$. If we define a dynamics on that set by
\be 
L_U(V)=UV
\ee
then we can readily verify that it is norm-preserving. Applying our existing result tells us that this dynamics is recurrent in metric, and so for any $\epsilon'$, there is some $n$ such that $\epsilon'>\tr((U(n)-1,U(n)-1)$. 

Now, for any operator $A$ and any vector $v$, $\tr(A^\dagger A)\langle v,v\rangle\geq \langle Av,Av$. And because the original state space is bounded, there is some $D$ such that $\langle s,s\rangle<D$ for any state $s$. So if we take $\epsilon'=\epsilon^2/ D$, we have
\[ 
\epsilon^2 > \tr(U(n)-1,U(n)-1) D > \tr(U(n)-1,U(n)-1)\langle s,s\rangle
\]
\be
\geq \langle U(n)s-s,U(n)s-s\rangle=d(U(n)s,s)^2.
\ee
$\Box$
\end{proof}
In fact, the result rests only on the existence of a preserved distance function, and on the fact that the space can be covered, for any $\epsilon$, with finitely many balls of radius  $\epsilon$. (This latter property is called \emph{total boundedness} in the theory of metric spaces.) Since the details are mildly fiddly, I postpone them to the appendix.

\section{Application: quantum mechanics}

In the quantum theory of pure states, quantum states are taken to be vectors of norm 1; in the theory of mixed states, they are taken to be self-adjoint operators of trace 1 all of whose eigenvalues are nonnegative. In both cases the states form a bounded subset of a linear space. So whenever the quantum theory's Hilbert space is finite-dimensional, the dynamics are uniformly recurrent.

In many cases of physical interest, though, we take the Hilbert space to be infinite-dimensional. In some cases the recurrence theorem can be extended to these situations also. Specifically, we can prove the
\begin{theorem}
\textbf{(Quantum recurrence theorem)} Given a quantum system with Hamiltonian $\op{H}$, if the set of eigenvalues of $\op{H}$ (a) is bounded below, (b) is discrete, and (c) has only finitely many elements within any bounded interval, then under the evolution operator $\op{U}=\exp(-i \op{H}\tau)$ (for arbitrary $\tau>0$),
\begin{enumerate}
\item[(i)]The dynamics are recurrent in metric;
\item[(ii)]For any $E$, the restriction of the dynamics to states of expected energy $<E$ is uniformly recurrent in metric.
\end{enumerate}
\end{theorem}
\begin{proof}
To begin with, let the energy eigenstates be given by $\ket{\varphi_i}$, with $\op{H}\ket{\varphi_i}=E_i\ket{\varphi_i}$, and with the energies in increasing order, so that if $i>j$, $E_i\geq E_j$. (That they can be so listed is a consequence of (a)-(c).) Write $\mc{H}_N$ for the Hilbert space spanned by states $\{\ket{\varphi_1},\ldots \ket{\varphi_N}\}$; this space has dimension $N$. $\op{U}$ leaves $\mc{H}_N$ invariant and so defines a dynamics on it; since $\mc{H}_N$ is finite-dimensional, this dynamics is uniformly recurrent.

Next, note that any given state can be expressed as a sum of energy eigenstates: $\ket{\psi}=\sum_i^\infty \alpha_i \ket{\varphi_i}$, where $\sum_i^\infty |\alpha_i|^2=1$. We can usefully define the partial sums $\ket{\psi_N}=\sum_i^N \alpha_i\ket{\varphi_i}$ (that is, as the projection of $\ket{\psi}$ onto $\mc{H}_N$, so that $\bk{\psi}{\psi_N}=\sum_i^N|\alpha_i|^N$. Using the linear space definition of distance, we can readily calculate 
\[
d(\ket{\psi},\ket{\psi}_N)=\sqrt{2\left(1-\sum_{i=0}^N|\alpha_i|^2\right)}=\sqrt{2\sum_{i=N+1}^\infty|\alpha_i|^2}.
\]

To prove (i), pick $\epsilon>0$. There will exist some $N$ such that $d(\ket{\psi},\ket{\psi_N})<\epsilon/3$. Since the dynamics on $\mc{H}_N$ are recurrent, there will also exist $m$ such that $d(\op{U}(m)\ket{\psi_N},\ket{\psi_N})<\epsilon/3$. By the triangle inequality,
\[
d(\op{U}(m)\ket{\psi},\ket{\psi})<d(\op{U}(m)\ket{\psi},\op{U}(m)\ket{\psi_N})\]
\be +d(\op{U}(m)\ket{\psi_N},\ket{\psi_N})+d(\ket{\psi_N},\ket{\psi}).
\ee
Since $\op{U}$ preserves Hilbert-space distance, the first and third terms are equal, and each is less than $\epsilon/3$; the result now follows.

To prove (ii), assume that all the energy eigenvalues are nonnegative (if this is not the case, we can just add a constant to $\op{H}$ without affecting the dynamics) and that there are eigenvalues of arbitrarily high energy (if not, then by (b) and (c), the system is finite-dimensional and so already known to be uniformly recurrent). Now pick $E>0$. For any $\lambda>0$,  there will exist some energy eigenvalue $E_N$ with $E_N>E/\lambda$. For any given state $\ket{\psi}$ with expected energy $<E$, we can decompose it as $\ket{\psi}=\ket{\psi_N}+\ket{\delta \psi}$, and we have 
\begin{eqnarray}
E & \geq & \matel{\psi}{\op{H}}{\psi} \\
& = & \matel{\psi_N}{\op{H}}{\psi_N}+\matel{\delta\psi}{\op{H}}{\delta\psi} \\
& \geq &\matel{\delta\psi}{\op{H}}{\delta\psi}\\
& = & \sum_{i=N+1}^\infty E_i |\alpha_i|^2\\
& \geq & E_N \sum_{i=N+1}^\infty |\alpha_i|^2 \\
& > &  (E/\lambda) \sum_{i=N+1}^\infty |\alpha_i|^2 \\
\end{eqnarray}
Rearranging,
\be 
2\lambda>d(\ket{\psi},\ket{\psi_N})^2.
\ee
Since $\lambda$ was chosen arbitrarily, we can in particular choose $N$ so that \be d(\ket{\psi},\ket{\psi_N})<\epsilon/3.\ee  (In other words, all states with expected energy below $E$ are within distance $\epsilon/3$ of $\mc{H}_N$). Since the dynamics are uniformly recurrent on $\mc{H}_N$, we can find $m$ such that for all states $\ket{\psi}$, $d(\op{U}(m)\ket{\psi_N},\ket{\psi_N})<\epsilon/3$. Repeating the argument in the proof of (i) tells us that $d(\op{U}(m)\ket{\psi}, \ket{\psi})<\epsilon$. $\Box$
\end{proof}

When will this somewhat abstract condition be realised? On heuristic grounds we would expect it to hold in nonrelativistic QM whenever the system is spatially bounded and has a potential bounded below: in this situation, a solution to the Schr\"{o}dinger equation can only have discretely many nodes, and any increase in nodes increases the average curvature of the solution and so the kinetic energy. In fact this can be formally proved in functional analysis (see \cite[p.1330]{dunfordschwartzII}). More generally, \cite{bekensteinbound} has advanced arguments to the effect that even in quantum field theory, there are only finitely many energy eigenstates of any spatially bounded theory below a given total energy.

If the theory is not spatially bounded, the spectrum may be continuous, and there is no guarantee of recurrence, just as in the classical case (the free particle is a trivial counterexample). 

\section{Another route to linear recurrence}

The above argument is not in fact how the quantum recurrence theorem was originally proved. In this section I sketch the standard proof, which has the advantage of making more transparent the way in which recurrence actually plays out, at the cost of blurring the continuity with the classical and discrete cases.

Firstly, suppose that the system's Hilbert space is finite-dimensional, and that the Hamiltonian has eigenvalues $E_1, \ldots E_N$ and eigenstates $\ket{\varphi_1},\ldots \ket{\varphi_N}$. Then any state $\ket{\psi}$ may be expanded in energy eigenstates, and we have
\be 
\op{U}(t)\ket{\psi}=\exp(-i \op{H}t)\ket{\psi}=\sum_{i=1}^N \alpha_i \e{-iE_i t}\ket{\varphi_i}.
\ee
Now suppose that all of the energy eigenvalues $E_1,\ldots E_N$ are rational numbers. Then for some time $t$, $E_i t = 2n_i \pi$ for some integers $n_i$. It follows that $\op{U}(t)\ket{\psi}=\ket{\psi}$: that is, the theory is uniformly recurrent.

But any real number is arbitrarily well approximated by a rational number! So this suggests that this result ought to hold (to an arbitrarily good approximation) for arbitrary $E_1,\ldots E_N$. This leads us to the theory of \emph{almost periodic functions}, and outside the scope of this article. But an elementary sketch proof (using our existing recurrence results!) would go as follows: for given $\tau$, the effect of $U(\tau)$ is to increase the phase of the $i$th eigenstate by $E_i \tau$. These phases take values between $0$ and $2\pi$; that is, the set of all phases lives on an $N$-dimensional torus. If we define a distance or volume function on that torus in the usual way, we can see that these phase increases preserve both. Applying the volume or metric versions of the recurrence theorem then tells us that to any degree of approximation, and for some $N$, applying the phase rotation $N$ times gets us back where we started to that degree of approximation. (This argument essentially follows \cite{Schulman1978}.)

\section{Two artefacts of classical mechanics}

In the literature on the foundations of statistical mechanics, two subtleties of the classical (Poincar\'{e}) recurrence theorem are often stressed:\footnote{See, \egc, \cite[pp.182--7,199--202]{sklarstatmech}, and references therein.}
\begin{enumerate}
\item Although the theorem establishes (roughly) that the states that do not recur have collective measure zero, nothing rules out the existence of anomalous states that never recur.
\item The theorem establishes recurrence, but not \emph{uniform} recurrence.
\end{enumerate}
The first point (often called the `problem of measure zero'), and its cousin in the theory of ergodicity, poses a supposed problem to any attempt to give a dynamical justification to the standard measures used in statistical mechanics: only by making measure-theoretic assumptions can we even get the recurrence theorem to tell us about actual systems. The second point is advanced as a supposed justification for the move from the Boltzmannian description of thermodynamical systems as represented by phase-space \emph{points}, to the Gibbsian description which represents them by phase-space \emph{distributions}. The idea here is that recurrence poses a threat to the monotonicity of entropy increase --- but if entropy is a function of distributions rather than individual states, there is no objection to its increasing indefinitely.

I make no comment on the \emph{conceptual} status of these concerns --- but we can see that technically speaking, they rely on features of classical dynamics which do not occur in quantum theory. The quantum recurrence theorem demonstrates \emph{uniform} recurrence, and it tells us that \emph{all} states return arbitrarily close to their starting point. So insofar as a given classical system fails to demonstrate these features, all we can learn is that the classical description must at some point fail to be a good approximation to any underlying quantum system. And since our only (non-historical) reason to describe a system via classical dynamics is because it is assumed to be a good approximation to some underlying quantum dynamics, these subtleties can be dismissed as classical artefacts, without conceptual significance.

(To be fair, at least the second point is sometimes noted. \cite{sklarstatmech}, in particular, notes (pp.200-1) that the Gibbsian approach does indeed fail in finite-volume quantum systems. But he suggests that the solution is to go to the thermodynamic limit: that is, to infinite-volume (but finite-density) systems. However, there is no reason to expect even \emph{non-uniform} recurrence to apply in that context, and  indeed, a move to the infinite-volume case would have saved classical \emph{Boltzmannian} statistical mechanics from any problem of recurrence. So it is hard to see what is really intended here.)

It is interesting to ask just what features of classical mechanics really prevent uniform recurrence. On possibility is the absence of a dynamically conserved distance function (let alone a dynamically preserved linear structure) on phase space. Classical dynamics preserves size but not shape\footnote{That is not to say that it preserves \emph{only} size: by no means are all volume-preserving maps symplectic maps.} and so allows a wider variety of possible dynamics.

However, the space of probability distributions over phase space does have a linear structure, and the dynamics on that space is likewise linear: the Liouville equation, $\dot{\rho}=\pb{\rho}{H}$, is a linear differential equation. Furthermore, if we restrict our attention to a finite-volume region of phase space and define an inner product on the space of density operators by 
\be 
\langle \rho,\sigma\rangle=\int \dr{x}\rho(x)\sigma(x)
\ee
(where the integral is taken with respect to Liouville measure) then it is easy to show that this inner product is preserved by the Liouville equation. In this case, the reason that classical dynamics fails to abide by the linear recurrence theorem is that the distribution space is infinite-dimensional, even if phase space has finite volume: distributions can have structure on arbitrarily short scales. 

Formally, quantum theory can also be done on the space of phase-space distributions: the Wigner function formalism\footnote{The Wigner function was first introduced in \cite{wignerfunction} and explored further by \cite{moyal}; see Hillery~\emph{et al}~\citeyear{hilleryetal} for a review of its properties.} maps each quantum state to a real function on phase space (with the proviso that such functions are not reliably nonnegative) and the resultant mapping of the Schr\"{o}dinger equation --- the \emph{Moyal equation} --- is equal to Liouville's equation to leading order in $\hbar$. But Winger functions of a given energy cannot vary on arbitrarily short lengthscales, essentially because of the uncertainty principle. We can again see the infinite dimensionality of the space of classical distributions as an artefact of classical physics.

\section{Conclusion}

Systems undergo recurrence, essentially, because there are only so many places they can go, and eventually they have to go back on themselves, and because if they have 1:1 dynamics, the only way they can do this is by returning to their starting point. Generalising this to state spaces with infinitely many states can be done via a conserved volume, or a conserved inner product, or a conserved metric. The first of these is most 
commonly seen, and applies to classical mechanics; the others, however, are more powerful, implying \emph{exceptionless}, \emph{uniform} recurrence, and are applicable to quantum mechanics. The relative weakness of the classical recurrence theorem is therefore a curiosity rather than a point of conceptual significance.

\appendix

\section*{Appendix: Recurrence in metric spaces}

To begin with some definitions: a \emph{pre-metric space} is a set $\mc{S}$ equipped with a function $d(x,y)$ from pairs of elements $x,y$ of the set to the nonnegative real numbers, such that for any $x,y,z\in \mc{S}$:
\begin{enumerate}
\item $d(x,x)=0$.
\item $d(x,y)=d(y,x)$;
\item $d(x,y,z)\leq d(x,y)+d(y,z)$ (the triangle inequality).
\end{enumerate}
(A \emph{metric space} satisfies in addition that if $d(x,y)=0$ then $x=y$, but we will not need this condition.)
A metric dynamical theory is then just a dynamical theory whose state space is equipped with a metric.

We now repeat some of the definitions used for linear spaces:
\begin{itemize}
\item Given a point $x$, and $\epsilon>0$, the \emph{ball of radius} $\epsilon$, $B_\epsilon (x)$ is the set $\{y\in\mc{S}:d(x,y)<\epsilon\}$.
\item A premetric space is \emph{totally bounded} if for any $\epsilon>0$, there are finitely many points $x_1,\ldots x_N$ such that
\[
B_\epsilon(x_1) \cup \cdots \cup B_\epsilon(x_N)=\mc{S}. 
\]
\item A map $U:\mc{S}\rightarrow \mc{S}$ is \emph{metric-preserving} if for any $x,y\in \mc{S}$, $d(x,y)=d(U(x),U(y))$.
\end{itemize}
The proof that any metric-preserving evolution map is recurrent in metric proceeds the same way as for linear spaces:
\begin{theorem}\label{metricrecurrence}
\textbf{(Metric recurrence theorem)} Any metric dynamical theory which is totally bounded and metric-preserving is recurrent in metric.
\end{theorem}
\begin{proof}
Fix $\epsilon>0$ and $x\in \mc{S}$ (where $\mc{S}$ is the state space), and let $\{x_1,\ldots x_N\}$ be such that the union of the $B_{\epsilon/2}(x_i)$ is $\mc{S}$.  Since there are only finitely many $x_i$, we can find $n,m (n<m)$ and $i$ such that $U(n)x$ and $U(m)x$ are both in $B_{\epsilon/2}(x_i)$, and so by the triangle inequality, $d(U(n)x,U(m)x)<\epsilon$. Since $U$ is metric-preserving, $d(x,U(m-n)x)<\epsilon$.
\end{proof}

To extend this to uniform recurrence, we can use the usual trick of applying recurrence to dynamics on the space of metric-preserving maps. To do so, we need to make this space into a totally bounded metric space. 

To do so, first note that a totally bounded space is also bounded: that is, for some $D>0$, $d(x,y)<D$ for all $x$ and $y$. (If the balls of radius (say) 1 around $x_1\ldots x_N$ contain all points in the space, then for some $D'>0$ and all $i,j$, $d(x_i,x_j)<D'$, and so by the triangle inequality, $d(x,y)<D'+2$.)
If $\mc{U}(\mc{S},\mc{T})$ is the space of metric-preserving functions from $\mc{S}$ to $\mc{T}$, and $\mc{S}$ and $\mc{T}$ are totally bounded metric spaces, we can then make $\mc{U}(\mc{S},\mc{T})$ into a metric space by defining 
\[
d(f,g)=\mathrm{max}_x d(f(x),g(x)).
\]
We can prove the following lemma:
\begin{description}
\item[Lemma:] If $\mc{S}$ and $\mc{T}$ are totally bounded premetric spaces (with metrics $d_S$, $d_T$), so is $\mc{U}(\mc{S},\mc{T})$ (with the metric $d$ defined above).
\end{description}
\begin{proof} Pick $\epsilon>0$, and pick $x_1\ldots x_N\in \mc{S}$, $y_1,\ldots y_M\in \mc{T}$ such that 
\[
B_{\epsilon/4}(x_1) \cup \cdots \cup B_{\epsilon/4}(x_N)=\mc{S} 
\]
\[
B_{\epsilon/4}(y_1) \cup \cdots \cup B_{\epsilon/4}(y_M)=\mc{T} .
\]
Now let $\Pi$ be the set of functions from $\{x_1,\ldots x_N\}$ to $\{y_1,\ldots y_M\}$, and for each $\pi\in\Pi$, define $F_\pi\in\mc{U}(\mc{S},\mc{T})$ to be the set of metric-preserving functions such that $f(x_i)\in B_{\epsilon/4}(\pi(x_i))$ for each $i$. Suppose $f$ and $g$ are both elements of some $F_\pi$, and that $x\in B_{\epsilon/4}(\pi(x_j))$. Then by the triangle inequality,
\[ 
d_T(f(x),g(x))\leq d_T(f(x),f(x_i)) + 
\]
\be
d_T(f(x_i),\pi(x_i))+d_T(\pi(x_i),g(x_i))+d_T(g(x_i),g(x)).
\ee
Since $f$ and $g$ are metric-preserving, $d_T(f(x),f(x_i))=d_S(x,x_i)<\epsilon/4$; since $f\in F_\pi$, $d_T(f(x_i),\pi(x_i))<\epsilon/4$; the same holds for $g$. So we conclude that $d_T(f(x),g(x))<\epsilon$.

Now for each non-empty $F_\pi$, pick an arbitrary $f_\pi\in F_\pi$. $F_\pi\subset B_\epsilon(f_\pi)$, so it follows that the union of all the  $B_\epsilon(f_\pi)$ is $\mc{U}(\mc{S},\mc{T})$. Since there are only finitely many of these, $\mc{U}(\mc{S},\mc{T})$ is totally bounded.
\end{proof}

In particular, $\mc{U}(\mc{S},\mc{S})$ is totally bounded. The dynamics defined by $L_U(f)=U \cdot f$ is a metric-preserving evolution on this space, and so by theorem \ref{metricrecurrence}, for any $\epsilon$ we can find $n$ such that $d(U(n)1,1)<\epsilon$, and hence $d_S(U(n)x,x)<\epsilon$ for any $x$. That is, we have proved
\begin{theorem}
\textbf{(Uniform metric recurrence theorem)} Any metric dynamical theory which is totally bounded and metric-preserving is uniformly recurrent in metric.
\end{theorem}
For readers more familiar with the topological notion of \emph{compactness}, note that any compact metric space is totally bounded\footnote{The converse is not true, but any \emph{complete} totally bounded metric space is compact. For an elementary proof, see \cite[p.141]{sutherland}.} (for given $\epsilon>0$, consider the open cover consisting of balls of radius $\epsilon$ around every point in the space) and so, as a corollary of theorem~\ref{metricrecurrence}, any metric-preserving dynamical theory with a compact state space is uniformly recurrent in metric.

\end{document}